%% file: antidistinguishability_arxiv_v2.tex
\documentclass[12pt,reqno]{amsart}
\usepackage{amsthm}
\usepackage{amsmath}
\usepackage{latexsym}
\usepackage{amsfonts}
\usepackage{amssymb}
\usepackage{color}
\usepackage{bbm,dsfont}
\usepackage{graphicx}
\usepackage{hyperref}
\usepackage{enumerate}
\usepackage{comment}
\usepackage[position=top]{subfig}
\usepackage{cite}

\newtheorem{proposition}{Proposition}

\newtheorem{corollary}{Corollary}

\theoremstyle{definition}

\newtheorem{example}{Example}
\newtheorem{definition}{Definition}

\newcommand{\C}{\mathbb{C}} 
\newcommand{\real}{\mathbb R} 
\newcommand{\complex}{\mathbb C}

\newcommand{\half}{\tfrac{1}{2}} 
 
\newcommand{\hi}{\mathcal{H}} 

\newcommand{\ip}[2]{\left\langle\,#1\,|\,#2\,\right\rangle} 
\newcommand{\no}[1]{\left\|#1\right\|} 
\newcommand{\tr}[1]{{\rm tr}\left[#1\right]} 

\newcommand{\id}{\mathbbm{1}} 

\renewcommand{\rho}{\varrho}

\newcommand{\vr}{\mathbf{r}} 
\newcommand{\vsigma}{\boldsymbol{\sigma}} 

\newcommand{\Mo}{\mathsf{M}}

\begin{document}\setlength{\arraycolsep}{2pt}

\title[]{Antidistinguishability of Pure Quantum States}

\author{Teiko Heinosaari}
\author{Oskari Kerppo}

\address{QTF Centre of Excellence, Turku Centre for Quantum Physics, Department of Physics and Astronomy, University of Turku, FI-20014 Turku, Finland}

\begin{abstract}
The Pusey-Barrett-Rudolph theorem has recently provoked a lot of discussion regarding the reality of the quantum state. In this article we focus on a property called antidistinguishability, which is a main component in constructing the proof for the PBR theorem. In particular we study algebraic conditions for a set of pure quantum states to be antidistinguishable, and a novel sufficient condition is presented. We also discuss a more general criterion which can be used to show that the sufficient condition is not necessary. Lastly, we consider how many quantum states needs to be added into a set of pure quantum states in order to make the set antidistinguishable. It is shown that in the case of qubit states the answer is one, while in the general but finite dimensional case the answer is at most $n$, where $n$ is the size of the original set.
\end{abstract}

\maketitle

\section{Introduction}

The reality of the quantum state is an important and debated topic in the foundations of quantum mechanics; either the quantum state is to be taken as a state of knowledge, or it is to be taken as a state of reality, possibly even some hybrid of these two.
The main distinction is thus between realistic and epistemic interpretations of the quantum state; see e.g. \cite{Spekkens05,Spekkens07,PuBaRu12,Leifer14,HaSp10,CoRe17,Ha13,LeMa13}.
 
Suppose that a pure quantum state $\varphi$ is prepared. In the ontological models framework it is understood that, whenever a quantum state  is prepared, an ontic state $\lambda$ from the ontic state space $\Lambda$ emerges as a result. The role of the quantum state is then to encode a probability distribution $\mu_\varphi$ over the ontic state space.
Suppose, then, that a measurement is performed on the quantum state. 
In the ontological models framework every measurement $M$ corresponds to a set of conditional probability distributions $\xi(x \,|\, M , \lambda)$, and the probability of an outcome $x$ is determined by \begin{equation}\label{1}
p(x \,|\, \varphi, M)= \int_\Lambda \xi(x \,|\, M , \lambda) \, d\mu_\varphi(\lambda)
\end{equation} The distinction between the ontic and epistemic ontological models is based on whether the probability measures corresponding to distinct pure states overlap or not \cite{Spekkens07,HaSp10}. 
An ontological model is \emph{$\psi$-ontic} if every pair of distinct pure states  correspond to non-overlapping probability measures. Otherwise the ontological model is \emph{$\psi$-epistemic}.

The overlap between the probability distributions related to distinct quantum states in an ontological model is one of the main features that characterizes the properties of the model. 
We see from \eqref{1} that, whenever $\varphi$ and $\psi$ are distinguishable (i.e. orthogonal) pure states, the intersection of the supports of $\mu_\varphi$ and $\mu_\psi$ has to be of measure zero. 
We can get more information about the supports by looking not only at distinguishable but also \emph{antidistinguishable} sets of pure states -- this relation on states is the topic of the current investigation.
A set of $n$ states is called antidistinguishable if there exists a measurement such that for each state in the set, there is at least one measurement outcome occuring with zero probability; see Fig \ref{fig:1} for an illustration. 
Antidistinguishable sets of quantum states have an important property in the ontological models framework. 
Indeed, if $\{ \rho_i \}$ is an antidistinguishable set, then it can be shown that there does not exist a set $\Omega \subset \Lambda$ such that $\mu_{\rho_i} (\Omega ) > 0$ for every $\rho_i$ \cite{Leifer14}. This means that, whenever a quantum state is sampled from $\{ \rho_i \}$, the emerging ontic state $\lambda$ cannot belong to the support of every $\mu_{\rho_i}$. If this was not the case, any measurement device performing an antidistinguishable measurement on an antidistinguishable set of quantum states would run the risk of giving a measurement outcome that is not compatible with quantum theory. To see this, suppose that the quantum state $\rho_j$ is prepared. In an antidistinguishing measurement, any ontic state $\lambda$ from the support of $\mu_{\rho_j}$ would lead to the measurement result $j$ with zero probability. If it was possible for $\lambda$ to belong to the support of every $\mu_{\rho_i}$, the measurement device would be unable to decide the result of the measurement.
This is exactly the property that is used in the proof of the influential PBR Theorem \cite{PuBaRu12}.

\begin{figure}
\centering
\begin{tabular}{ll}
\textit{a)} & \raisebox{-.35\height}{\includegraphics[scale=.7]{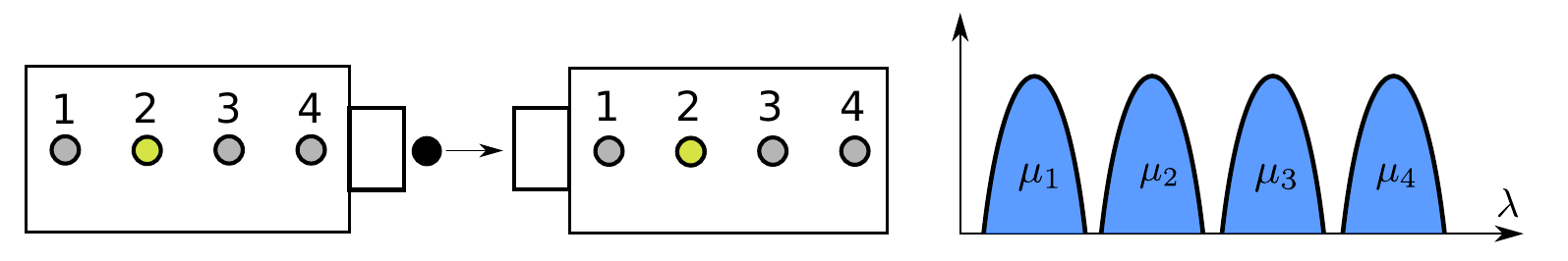}}
\end{tabular}
\begin{tabular}{ll}
\textit{b)} & \raisebox{-.35\height}{\includegraphics[scale=.7]{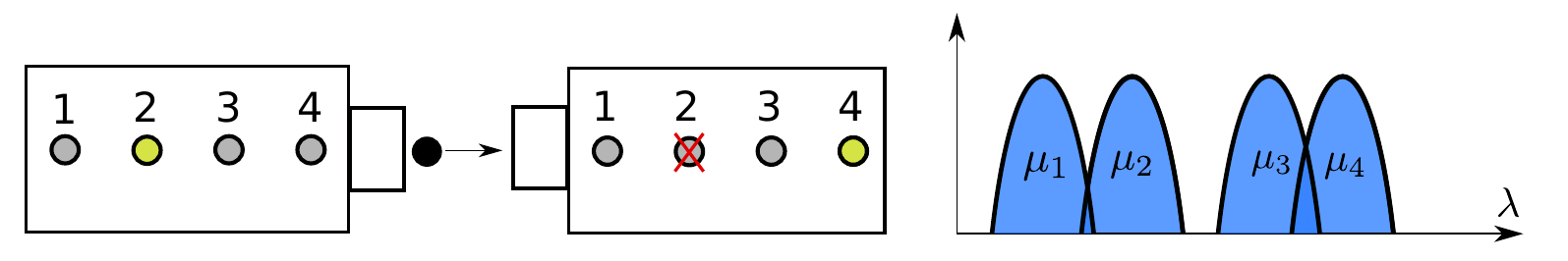}}
\end{tabular}

\captionsetup{width=\linewidth}
\caption{\textit{a)} Preparation and measurement of a state in the distinguishable case. There is no overlap between corresponding probability distributions. \textit{b)} In the antidistinguishable case there exists a measurement outcome with zero probability for every possible preparation of a state. There may be overlap between corresponding probability distributions.}\label{fig:1}
\end{figure}

While PBR used antidistinguishability to prove their main result, antidistinguishability is not a new concept. It has gone by the names of post-Peierls incompatibility \cite{CaFuSc02} and conclusive exclusion \cite{BaJaOpPe14} in past literature. 
Some results regarding antidistinguishable sets are known. 
The antidistinguishability of qubit states was characterized in \cite{CaFuSc02}.
It was also shown, for instance, that three rank-two states in three dimensional Hilbert space are antidistinguishable if and only if the vectors orthogonal to their supports are mutually orthogonal.
Another result, proved in \cite{BaJaOpPe14}, states that a set of $n$ quantum states is antidistinguishable only if 
\begin{align}\label{fidelitycondition}
\sum_{j\neq k} F(\rho_j, \rho_k ) \leq n(n-2) \, , 
\end{align}
where $F(\rho_j, \rho_k )$ is the fidelity between the quantum states $\rho_j$ and $\rho_k$. 
Further, a numerical quantification of antidistinguishability has been used in \cite{Knee17} in order to find an optimal setup to test the reality of the quantum state.

In this work, we present a sufficient condition for a finite set of pure states to be antidistinguishable. 
We show that this condition reduces to the known \cite{CaFuSc02} necessary and sufficient condition in the case of pure qubit states.
We demonstrate that, in higher dimensions, the condition is not necessary, and we discuss a more general criterion, which is necessary and sufficient but not explicit.
Both of our conditions are algebraic in nature, and they do not only give a condition for antidistinguishability but they also give a concrete form for an antidistinguishing measurement.

It is interesting to observe that a non-antidistinguishable set of states may become antidistinguishable when we add states to it. 
We will discuss the minimal number of states that has to be added in order to make any finite set of states antidistinguishable.
In particular, we prove that any finite set of pure qubit states is either antidistinguishable or can be made antidistinguishable by adding a single pure state. 

\section{Definitions and basic properties}

Quantum measurements are generally described by positive operator valued measures (POVMs). 
We recall that a POVM (with finite number of outcomes) is a map $\Mo:j \mapsto \Mo(j)$ that assigns a positive operator to each measurement outcome and is normalized, $\sum_j \Mo(j) = \id$. 

\begin{definition}
States $\varrho_1,\ldots,\varrho_n$ are
\begin{itemize}
\item[(a)] \emph{distinguishable} if there exists a POVM $\Mo$ with the outcomes $1,\ldots,n$ such that $\tr{\varrho_j \Mo(j)}=1$ for every $j=1,\ldots,n$.  
\item[(b)] \emph{antidistinguishable} if there exists a POVM $\Mo$ with the outcomes $1,\ldots,n$ such that $\tr{\varrho_j \Mo(j)}=0$ and $\sum_{k=1}^n \tr{\varrho_k \Mo(j)} \neq 0$ for every $j=1,\ldots,n$.
\end{itemize}
\end{definition}

The latter condition in (b) was not stated in the definition given in \cite{Leifer14}. 
We add it here as an additional requirement to avoid the trivial case where some of the outcomes of the POVM never occur, and it seems to us that this condition is implicitly used in \cite{Leifer14} and \cite{CaFuSc02}.
The additional requirement means that every measurement outcome occurs at least sometimes with the given set of states.
Without this requirement we could have e.g. $\Mo(j)=0$ for some outcomes $j$, and then any set containing two antidistinguishable states would be antidistinguishable.

The following two facts were observed in \cite{Leifer14}.

\begin{proposition}
\begin{itemize}
\item[(a)] If states $\varrho_1,\ldots,\varrho_n$ are distinguishable, then they are also antidistinguishable.  
\item[(b)] Two states $\varrho_1$ and $\varrho_2$ are distinguishable if and only if they are antidistinguishable.
\end{itemize}
\end{proposition}

These statements are easily seen to be true.
Namely, if the states $\varrho_1,\ldots,\varrho_n$ are distinguishable, then they are also antidistinguishable because we can just rearrange the ouctomes of $\Mo$. If two states are antidistinguishable then they are also distinguishable because the exclusion of one state implies that the other occurred with certainty.

It is useful to note that the antidistinguishability condition can be written without the trace.
Namely, $\tr{\varrho_j \Mo(j)}=0$ if and only if $\varrho_j\Mo(j)=0$.
(To see that the first condition leads to the second one, we can write $\tr{\varrho_j \Mo(j)} = \tr{C^*C}$ where $C=\sqrt{\varrho_j}\sqrt{\Mo(j)}$ and $C^*$ is the adjoint (or Hermitian conjugate) of $C$. 
Then $\tr{C^*C}=0$ implies that $C=0$, from which we get $\varrho_j\Mo(j)=0$.)
As we have ruled out the possibility of having $\Mo(j)=0$, we see that \emph{a full-rank (i.e. invertible) state is not antidistinguishable with any other states.}
 
An essential feature that is true for distinguishability but not for antidistinguishability is the following. 
If a set $\mathcal{A} = \{\varrho_1,\ldots,\varrho_n\}$ is distinguishable, then any subset of $\mathcal{A}$ is also distinguishable.
This is, perhaps surprisingly, not the case for an antidistinguishable set; even if a set of states $\mathcal{A}$ would not be antidistinguishable, it may be possible to find an antidistinguishable set $\mathcal{B}$ that contains $\mathcal{A}$.
On the contrary, a feature that is true for antidistinguishability but not for distinguishability is the following.

\begin{proposition}\label{prop:union}
Let $\mathcal{A}= \{\varrho_1,\ldots,\varrho_n\}$ and $\mathcal{B}=\{\eta_1,\ldots,\eta_m\}$ be antidistinguishable sets and $\mathcal{A}\cap\mathcal{B} = \emptyset$. Then $\mathcal{A}\cup\mathcal{B}$ is antidistinguishable. 
\end{proposition}

\begin{proof}
Let $\Mo_1$ and $\Mo_2$ be antidistinguishable POVMs for $\mathcal{A}$ and $\mathcal{B}$, respectively. 
We define a POVM $\Mo$ with $n+m$ outcomes as
\begin{equation*}
\Mo(j) = \left\{ 
\begin{array}{ll}
\half \Mo_1(j) & \textrm{if } j=1,\ldots,n \\
\half \Mo_2(j) & \textrm{if } j=n+1,\ldots,m
\end{array} \right.
\end{equation*}
It is straightforward to verify that $\Mo$ is antidistinguishing $\mathcal{A}\cup\mathcal{B}$.
\end{proof}

\begin{proposition}\label{prop:2n}
Any finite set of $n$ pure states is a subset of an antidistinguishable set of at most $2n$ elements.
\end{proposition}

\begin{proof}
Let us consider $n$ pure states $P_1,\ldots,P_n$.
For each $i=1,\ldots,n$, we set $\varrho_i = \tfrac{1}{d-1} (\id - P_i)$ and $\mathcal{A}_i = \{P_i,\varrho_i\}$.
The states $P_i$ and $\varrho_i$ have disjoint supports and are therefore distinguishable, hence also antidistinguishable.
By Prop. \ref{prop:union} the finite union $\cup_i  \mathcal{A}_i$ is antidistinguishable. 
\end{proof}

In the following investigation we will focus on pure states. 
We will denote by $P$ a pure state, i.e., a one-dimensional projection. 
For any projection $R$, we denote $R^\perp \equiv \id - R$.

\section{Sufficient condition for the antidistinguishability of pure states}

Pure states $P_1,\ldots,P_n$ are distinguishable if and only if they are pairwise orthogonal, i.e., $P_jP_k=0$ for $j\neq k$.
This is equivalent to the condition that the sum $\sum_{j=1}^n P_j$ is a projection. 
In that case, the sum is the projection with the kernel $\cap_j \ker P_j$.

The following sufficient condition for antidistinguishability can be seen, from the mathematical point of view, as a natural generalization of the previously expressed condition for distinguishability.

\begin{proposition}\label{prop:sum=R}
Let $P_1,\ldots,P_n$ be pure states (and $n\geq 2$). 
\begin{itemize}
\item[(a)] If the sum $\sum_{j=1}^n t_jP_j$ is a projection for some positive real numbers $t_1,\ldots,t_n$, then it is the projection with the kernel $\cap_j \ker P_j$.
\item[(b)] Let us denote by $R$ the projection with the kernel $\cap_j \ker P_j$. If there are positive real numbers $t_1,\ldots,t_n$ such that 
\begin{align}\label{eq:sum=R}
\sum_{j=1}^n t_jP_j = R \, , 
\end{align}
then $P_1,\ldots,P_n$ are antidistinguishable. 
\end{itemize}
\end{proposition}

\begin{proof}
\begin{itemize}
\item[(a)] Assume that $\sum_{j=1}^n t_jP_j=:Q$ is a projection for $t_1,\ldots,t_n$. If $\varphi\in\hi$ is such that $P_j\varphi=0$ for every $j=1,\ldots,n$, then also $Q\varphi=0$.
Further, if $\psi \in \hi$ is such that $\ip{\psi}{P_j \psi} \neq 0$ for some $j=1,\ldots,n$, then also $\ip{\psi}{Q\psi} \neq 0$.
We conclude that the kernel of $Q$ is exactly the intersection of the kernels of $P_j$, i.e., 
\begin{equation}
\ker Q = \cap_j \ker P_j \, .
\end{equation}
Since $Q$ is a projection, this condition determines it completely.
\item[(b)] Assume that $t_1,\ldots,t_n$ and $R$ satisfy \eqref{eq:sum=R}.
Taking trace on both sides shows that $\sum_j t_j =\tr{R} \equiv r$. 
Since $n\geq 2$ and $P_j$ are different, $R$ cannot be one-dimensional and hence $r\geq 2$. 
We define 
\begin{equation}\label{eq:povm}
\Mo(j)=\tfrac{t_j}{r-1} (R - P_j) + \tfrac{1}{n}R^\perp 
\end{equation}
 for each $j=1,\ldots,n$.
It follows from \eqref{eq:sum=R} that $\sum_j \Mo(j) = \id$.
Further, \eqref{eq:sum=R} implies that $t_jP_j \leq R$, and as $P_j$ and $R$ are projections, this gives $P_j \leq R$.
We conclude that $\Mo(j)\geq 0$ and hence $\Mo$ is a POVM.

From $P_j \leq R$ follows that $P_j R =RP_j = P_j$ and $P_j R^\perp =R^\perp P_j = 0$. 
We thus have $P_j\Mo(j)=0$ and 
\begin{equation}
\sum_k \tr{P_k\Mo(j)}=\tfrac{t_j}{r-1}\sum_k (1- \tr{P_kP_j}) \neq 0
\end{equation}
as $\tr{P_kP_j}<1$ when $P_k \neq P_j$.
Therefore, the pure states $P_1,\ldots,P_n$ are antidistinguishable. 
\end{itemize}
\end{proof}

We remark that Eq. \eqref{eq:povm} in the proof of Prop. \ref{prop:sum=R} provides an explicit formula for the antidistinguishing POVM if the condition \eqref{eq:sum=R} is satisfied.                                                                                                                                                                                                                                                                                                                                                                                                                                                                                                                                                                                                                                                                                                                                                                                                                                                                                                                                                                                                                                                                                                                                                                                                                                                                                                                                                                                                                                                                                                                                                                                                                                                                                                                                                                                                                                                                                                                                                                                                                                                                                                                                                                                                                                                                                                                                                                                                                                                                                                                                                                                                                                                                                                                                                                                                                                                                                                                                                                                                                                                                                                                                                                                                                                                                                                                                                                                                                                                                                                                                                                                                                                                                                                                                                                                                                                                                                                                                                                                                                                                                                                           
For clarity, we also record the following immediate consequence of Prop. \ref{prop:sum=R}. 

\begin{corollary}\label{cor:id}
Let $P_1,\ldots,P_n$ be pure states. 
If there are positive real numbers $t_1,\ldots,t_n$ such that 
\begin{align}\label{eq:sum=id}
\sum_{j=1}^n t_jP_j = \id \, , 
\end{align}
then $P_1,\ldots,P_n$ are antidistinguishable. 
\end{corollary}

Let $P_1,\ldots,P_n$, $R$, and $t_1,\ldots,t_n$ be as in Prop. \ref{prop:sum=R}. The largest eigenvalue of the sum $\sum_{j=1}^n t_jP_j$ is between $\max_j t_j$ and $\sum_j t_j$.  
Since the largest eigenvalue of the projection $R$ is 1, it follows that 
\begin{equation}
\max_j t_j \leq 1 \leq \sum_j t_j \, . 
\end{equation}
In particular, we can limit to real numbers  $0<t_j \leq 1$ in the statement of Prop. \ref{prop:sum=R}.
Furthermore, the projection $R$, fixed by the pure states $P_1,\ldots,P_n$, essentially fixes the numbers $t_1,\dots ,t_n$. 
Namely, starting from \eqref{eq:sum=R}, we multiply the both sides by $P_k$ and take the trace. 
This yields a system of linear equations,
\begin{equation}\label{eq:sum=p_jk}
\sum_{k=1}^n p_{jk} t_k = 1 \, , \qquad j=1,\dots, n \, , 
\end{equation}
where $p_{jk} = \tr{P_j P_k}$. 
We can thus calculate the candidate numbers $t_1,\ldots,t_n$, which then still have to be tested whether they satisfy \eqref{eq:sum=R} or not.

\begin{example}
Let 
\begin{equation*}
P_1 = \begin{pmatrix}
1 &0&0\\0&0&0\\0&0&0
\end{pmatrix} \, , \quad P_2 = \tfrac{1}{5} \begin{pmatrix}
1 & 2 & 0 \\ 2 & 4 & 0 \\ 0&0&0
\end{pmatrix} \, , \quad P_3= \tfrac{1}{5} \begin{pmatrix}\hspace{-0.1cm}
\phantom{-}1 & -2&\phantom{-}0 \\\hspace{-0.1cm} -2 & \phantom{-}4 & \phantom{-}0 \\\hspace{-0.1cm} \phantom{-}0&\phantom{-}0&\phantom{-}0
\end{pmatrix}.
\end{equation*}
 It is straightforward to see that the projection with the kernel $\cap_i \ker P_i$ is 
 \begin{equation*}
 R = \begin{pmatrix}
1 &0&0\\0&1&0\\0&0&0
\end{pmatrix} \, .
\end{equation*}
 Solving Eq. \eqref{eq:sum=p_jk} then yields $t_1 = \frac34$, $t_2 = \frac 58$ and $t_3 = \frac58$.
 Using these numbers we get $t_1P_1 + t_2P_2 + t_3P_3=R$, therefore the pure states $P_1$, $P_2$, $P_3$ are antidistinguishable. 
The antidistinguishing POVM given in Eq. \eqref{eq:povm} is 
\begin{equation*}
\Mo (1) = \tfrac{1}{12} \begin{pmatrix}
0&0&0\\
0& 9&0\\
0&0& 4
\end{pmatrix}  , \enskip \Mo (2) = \begin{pmatrix}\hspace{-0.1cm}
\phantom{-}\frac12 & -\frac14 & \phantom{-}0 \\[0.3em]\hspace{-0.1cm} -\frac14 & \phantom{-}\frac 18 & \phantom{-}0\\[0.3em]\hspace{-0.1cm}\phantom{-}0&\phantom{-}0&\phantom{-}\frac13
\end{pmatrix}  , \enskip \Mo (3) = \begin{pmatrix}
\frac12 & \phantom{-}\frac14 & \phantom{-}0 \\[0.3em] \frac14 & \phantom{-}\frac 18 & \phantom{-}0\\[0.3em]0&\phantom{-}0&\phantom{-}\frac13
\end{pmatrix}.
 \end{equation*}

\end{example}

\section{Generating antidistinguishable sets of pure states}

In this section we explain and demonstrate how Prop. \ref{prop:sum=R} and Cor. \ref{cor:id} give an easy group theoretical method to construct antidistinguishable pure states.

Let $G$ be a finite group and $U:g \mapsto U_g$ its irreducible unitary representation on $\complex^d$. 
For a pure state $P$ and $g\in G$, we denote $P_g = U_g P U_g^*$. 
Hence, the group $G$ acts on the set of pure states. 
In the following, we take the pure state $P$ to be fixed and consider the orbit of $P$ under $G$, i.e., the set $\{P_g: g\in G\}$.
We assume that $P$ is chosen such that the orbit contains more than one element, i.e., that $P$ is not a fixed point. 
The operator $\sum_h P_h$ commutes with every $U_g$, and it hence follows from Schur's lemma that  $\sum_h P_h = c \id$ for some constant $c\in \complex$. 
Taking trace on both sides of this equation shows that $c=\# G/d > 0$, where $\# G$ is the order of $G$. 

Firstly, if $P_g \neq P_{h}$ whenever $g \neq h$, we conclude that the pure states  $P_g, g \in G$, are antidistinguishable as the condition \eqref{eq:sum=id} is valid for $t_j = d/\# G$. 
The POVM given in \eqref{eq:povm} is in this case
\begin{equation}
\Mo(g)=\tfrac{d}{\# G (d-1)} U_g P^\perp U_g^* \, .
\end{equation}
This is a special type of $U$-covariant POVM. 

Secondly, it is possible that $P_g = P_h$ for some $g \neq h$.
Let us denote by $\mathcal{I}(Q)$ the subgroup of $G$ that keeps a pure state $Q$ invariant, i.e., $\mathcal{I}(Q)=\{ h \in G: U_h Q U_h^* = Q \}$.
We have $\mathcal{I}(P_g) = \mathcal{I}(P)$ for every $g\in G$. 
The elements in the orbit of $P$ under $G$ can be labeled by the elements of the quotient group $G/\mathcal{I}(P)$, and for every $\gamma\in G/\mathcal{I}(P)$ we denote $P_\gamma = P_g$, where $g\in G$ is such that $g\cdot \mathcal{I}(P) = \gamma$.
Thereby, we have
\begin{equation}
\# \mathcal{I}(P) \cdot \sum_{\gamma \in G/\mathcal{I}(P)} P_q = \sum_{g \in G} P_g = \tfrac{\# G}{d} \id \, , 
\end{equation}
implying that the pure states  $P_\gamma, \gamma \in G/\mathcal{I}(P)$, are antidistinguishable as the condition \eqref{eq:sum=id} is valid for $t_j = \# \mathcal{I}(P)d/\# G$.

\begin{example}
The quaternionic group $Q$ consists of $8$ elements $\pm 1$, $\pm i$, $\pm j$, $\pm k$ satisfying the relations
\begin{equation*}
\begin{split}
(-1)^2 & = 1 \qquad (\pm 1) g = g(\pm 1)= \pm g \quad \forall g\in Q \\
i^2 & = j^2 = k^2 = -1 \qquad ij = -ji = k \, .
\end{split}
\end{equation*}
Let $\sigma_x$, $\sigma_y$, $\sigma_z$ be the usual Pauli matrices. 
The following map $U$ is an irreducible representation of $Q$,
\begin{align*}
U(\pm 1) = \pm \id \, , \quad U(\pm i) = \pm i \sigma_x \, , \quad U(\pm j) = \mp i \sigma_y \, , \quad  U(\pm k) = \pm i \sigma_z \, .
\end{align*}
We choose $P= \half (\id + \tfrac{1}{\sqrt{3}} (\sigma_x + \sigma_y + \sigma_z))$.
We then have $\mathcal{I}(P) = \{ \pm 1\}$, and hence the quotient group $Q/\mathcal{I}(P)$ contains four elements, $\bar{1}\equiv \{ \pm 1\}$, $\bar{i}\equiv\{ \pm i\}$, $\bar{j}\equiv\{ \pm j\}$, and $\bar{k}\equiv \{ \pm k\}$.
The corresponding pure states are 
\begin{align*}
& P_{\bar{1}} = \half (\id + \tfrac{1}{\sqrt{3}} (\sigma_x + \sigma_y + \sigma_z)) \, , & P_{\bar{i}} = \half (\id + \tfrac{1}{\sqrt{3}} (\sigma_x - \sigma_y - \sigma_z)) \\
& P_{\bar{j}} = \half (\id + \tfrac{1}{\sqrt{3}} (-\sigma_x + \sigma_y - \sigma_z)) \, , & P_{\bar{k}} = \half (\id + \tfrac{1}{\sqrt{3}} (-\sigma_x - \sigma_y + \sigma_z)) \, .
\end{align*}
Since these four states are generated following the previous method, they are antidistinguishable.  
It is also easily checked that, indeed, $P_{\bar{1}}+P_{\bar{i}}+P_{\bar{j}}+P_{\bar{k}}=2\id$.\end{example}

Finally, we remark that it may be more convenient to start from a \emph{reducible} unitary representation of a group $G$ on $\complex^d$, and then take an invariant subspace $\mathcal{V}\subset \complex^d$ such that the restriction $U \upharpoonright \mathcal{V}$ is irreducible. 
If we then choose a pure state $P$ such that the eigenspace of $P$ is a subspace of $\mathcal{V}$, our earlier method, applied now to the representation $U \upharpoonright \mathcal{V}$, still works. 
The only difference is that Schur's lemma leads to the identity operator $\id_{\mathcal{V}}$ on $\mathcal{V}$, but this can be seen as a projection on the original vector space $\C^d$. Therefore Prop. \ref{prop:sum=R} applies.
This slight generalization of the earlier method is demonstrated in the following example. 

\begin{example}
Let $S_n$ be the symmetric group of order $n$, $n\geq 3$.
Fix an orthonormal basis $\{\phi_1,\ldots,\phi_n\}$ for $\C^n$.
In the permutation representation of $S_n$ each element $g\in S_n$ permutes the coordinates in the fixed basis. 
The permutation representation is reducible, splitting into invariant subspaces of dimensions $1$ and $n-1$. 
The first one is generated by the vector $\sum_{j=1}^n \phi_j$, while the latter is given as $\{ \sum_{j=1}^n c_j \phi_j  : \sum_{j=1}^n c_j = 0 \} =: \mathcal{V}$. 
The restriction of the permutation representation to this $(n-1)$ -dimensional subspace is called the standard representation of $S_n$.
For instance, the following vectors belong to $\mathcal{V}$ and are on the same orbit in the standard representation of $S_3$,
$$
\psi_1 = \frac{1}{\sqrt{2}} \left(\!\!\! \begin{array}{c}\phantom{-}1 \\ -1 \\ \phantom{-}0
\end{array}\right) \, , \quad 
\psi_2 = \frac{1}{\sqrt{2}} 
\begin{pmatrix}
1\\ 0 \\ \!\!-1
\end{pmatrix} \, , \quad 
\psi_3 = \frac{1}{\sqrt{2}} 
\begin{pmatrix}
0\\ 1 \\ \!\!-1
\end{pmatrix}  \, .
$$
There are also three other vectors in this orbit, but these are just scalar multiples of the previous vectors and therefore determine the same pure states. 
We conclude that the following three pure states are antidistinguishable:
$$
P_1 = \frac{1}{2} \begin{pmatrix}
\hspace{-0.1cm}\phantom{-}1 & -1&\phantom{-}0\\\hspace{-0.1cm}-1&\phantom{-}1&\phantom{-}0\\\hspace{-0.1cm}\phantom{-}0&\phantom{-}0&\phantom{-}0
\end{pmatrix}  , \enskip P_2 = \frac{1}{2} \begin{pmatrix}
\hspace{-0.1cm}\phantom{-}1 & \phantom{-}0 & -1 \\\hspace{-0.1cm} \phantom{-}0 & \phantom{-}0 & \phantom{-}0 \\\hspace{-0.1cm} -1&\phantom{-}0&\phantom{-}1
\end{pmatrix}  , \enskip P_3= \frac{1}{2} \begin{pmatrix}
0&\phantom{-}0&\phantom{-}0\\0& \phantom{-}1 & -1 \\0& -1 & \phantom{-}1
\end{pmatrix} \, .
$$
From the previous method we know that the sum $P_1+P_2+P_3$, divided by $3/2$, is rank 2 projection; this is also easy to check directly. 
\end{example}

\section{General criterion}

Let $\{P_j\}_{j=1}^n$ be an antidistinguishable set of $n$ pure quantum states in a $d$-dimensional Hilbert space. 
There hence exists a POVM $\Mo$ such that $P_j \Mo(j)=0$ for each $j$.
It follows that the spectral decomposition of $\Mo(j)$ is 
\begin{equation}
\Mo(j) = \sum_{k=2}^{d} \alpha_j^k P_j^k \, , 
\end{equation}
where $\alpha_j^k\in [0,1]$ and $P_j^k$ are (pairwise orthogonal) one-dimensional projections satisfying
\begin{equation}\label{eq:id1}
P_j + \sum_{k=2}^{d} P_j^k = \id \, .
\end{equation}
Since $\Mo$ is a POVM, we also have
\begin{equation}\label{eq:id2}
\sum_{j=1}^n \sum_{k=2}^{d} \alpha_j^k P_j^k  = \id \, .
\end{equation}
Further, the condition $\sum_{i=1}^n \tr{P_i \Mo(j)}\neq 0$ means that
\begin{equation}\label{eq:id3}
\sum_{i=1}^n \sum_{k=2}^{d} \alpha_j^k \ \tr{P_i P_j^k} \neq 0 \, .
\end{equation}
In conclusion, from the antidistinguishability of $\{P_j\}_{j=1}^n$ follows that there exist one-dimensional projections $P_j^k$ such that \eqref{eq:id1}, \eqref{eq:id2} and \eqref{eq:id3} hold.
It is convenient to depict these projections in a chart; see Table \ref{table 1}.

\begin{table}
\centering
\begin{tabular}{lccr}
$P_1$&$P_2$&$\cdots$&$P_n$\\[3pt] \hline \\[-5pt]
$P^2_1$&$P^2_2$&$\cdots$&$P^2_n$\\
$\vdots$&$\vdots$&$\ddots$&$\vdots$\\
$P^{d}_1$&$P^{d}_2$&$\cdots$&$P^{d}_n$
\end{tabular}
\caption{Illustration of the projections used in Prop. \ref{generalprop}. 
Each column in the chart sums to the identity operator $\id$.}
\label{table 1}
\end{table}

We now record the following construction as a necessary and sufficient condition for the set of pure states to be antidistinguishable. 

\begin{proposition}\label{generalprop}
A set $\{P_j\}_{j=1}^n$ of $n$ pure quantum states in a $d$-dimensional Hilbert space is antidistinguishable if and only if there exist $(d-1)\cdot n$ pure states $P_j^k$ and real numbers $0\leq \alpha_j^k \leq 1$, $i\in\{1,\ldots,n\}$, $j\in\{2,\ldots,d\}$, such that \eqref{eq:id1}, \eqref{eq:id2} and \eqref{eq:id3} hold.
\end{proposition}

\begin{proof}
We have already seen that the `only if' holds. 
To prove the `if' part, we reverse the argument. 
Let us assume that pure states $P_j^k$ and real numbers $\alpha_j^k$ with the required properties \eqref{eq:id1}, \eqref{eq:id2} and \eqref{eq:id3} exist.
From \eqref{eq:id1} follows that we can define a POVM $\Mo$ as $\Mo(j) = \sum_{k=2}^{d} \alpha_j^k P_j^k$.
Then, \eqref{eq:id2} implies that $P_j$ is orthogonal to all $P_j^k$, therefore $P_j\Mo(j)=0$.
Finally, \eqref{eq:id3} implies that $\sum_{i=1}^n \tr{P_i \Mo(j)} \neq 0$. 
\end{proof}

We emphasize that while considering a set of $n$ pure quantum states one is not guaranteed to find an antidistinguishing POVM by completing the chart as shown in Table \ref{table 1}, even if the set under consideration is antidistinguishable. 
Prop. \ref{generalprop} merely guarantees that whenever the set is antidistinguishable, it is then possible to find the antidistinguishing POVM in the explained way. 
In the following example we demonstrate how the method works. 
This example also shows that the condition of Prop. \ref{prop:sum=R} is not necessary for a set of pure states to be antidistinguishable.

\begin{example}
Let us consider three pure states in dimension 3, given as
$$
P_1 = \begin{pmatrix}
1 &0&0\\0&0&0\\0&0&0
\end{pmatrix} \, , \quad P_2 = \tfrac{1}{5} \begin{pmatrix}
1 & 2 & 0 \\ 2 & 4 & 0 \\ 0&0&0
\end{pmatrix} \, , \quad P_3= \frac{1}{5} \begin{pmatrix}
0&0&0\\0& 1 & 2 \\ 0& 2 & 4
\end{pmatrix} \, .
$$
There are no positive real numbers $t_1$, $t_2$ and $t_3$ such that $\sum_{i=1}^3t_i P_i$ would sum up to a projection. 
One way to see this is to solve the $t$'s from Eq. \eqref{eq:sum=p_jk} and then find that the operator $R=\sum_{i=1}^3t_i P_i$ is not a projection, i.e., $R \neq R^2$. Another way is to solve the system of equations obtained from $\sum_{i=1}^3t_i P_i = \left(\sum_{i=1}^3t_i P_i\right)^2$. There it is evident that a solution where all of the $t$'s are nonzero and positive does not exist.
We thus conclude that Prop. \ref{prop:sum=R} is not applicable. 
However, a possible chart as in Prop. \ref{generalprop} is the following. 
\[\begin{array}{lcr}
\begin{pmatrix}
1 &0&0\\0&0&0\\0&0&0
\end{pmatrix}  &   \frac 15\begin{pmatrix}
1 & 2 & 0 \\ 2 & 4 & 0 \\ 0&0&0
\end{pmatrix} &  \frac15\begin{pmatrix}
0&0&0\\0&1 & 2 \\ 0&2&4
\end{pmatrix}\hspace{0.6cm}\\[0.7cm] \hline \\
\begin{pmatrix}
0&0&0\\0&0 & 0 \\ 0&0&1
\end{pmatrix} & \hspace{0.3cm}\begin{pmatrix}
0&0&0\\0&0 & 0 \\ 0&0&1
\end{pmatrix} &  
 \frac 15\begin{pmatrix}
0&\phantom{-}0&\phantom{-}0\\0&\phantom{-}4 & -2 \\0&-2&\phantom{-}1
\end{pmatrix}\hspace{0.2cm} \\ \\
\begin{pmatrix}
0&0&0\\0&1 & 0 \\ 0&0&0
\end{pmatrix} & \hspace{0.3cm}  \frac 15\begin{pmatrix}\hspace{-0.1cm}
\phantom{-}4&-2&\phantom{-}0\\\hspace{-0.1cm}-2&\phantom{-}1 & \phantom{-}0 \\ \hspace{-0.1cm}\phantom{-}0&\phantom{-}0&\phantom{-}0
\end{pmatrix} &   \begin{pmatrix}
1&0&0\\0&0 & 0 \\ 0&0&0
\end{pmatrix}\hspace{0.5cm}
\end{array}\]
We now find that $P_1^3+P_2^2+P_3^3 = \id$, so choosing 
\begin{equation}
[\alpha_j^k] = \begin{pmatrix}
0 & 1 & 0 \\
1 & 0 & 1
\end{pmatrix}
\end{equation}
we confirm that $P_1$, $P_2$ and $P_3$ are antidistinguishable.
In conclusion, the condition of Prop. \ref{prop:sum=R} is not necessary for pure states to be antidistinguishable. 
\end{example}

\section{Antidistinguishability of qubit states}

We have earlier observed that a full-rank state is not antidistinguishable with any other states.
In the case of qubit states, this means that only pure qubit states can be antidistinguishable. 

\begin{figure}
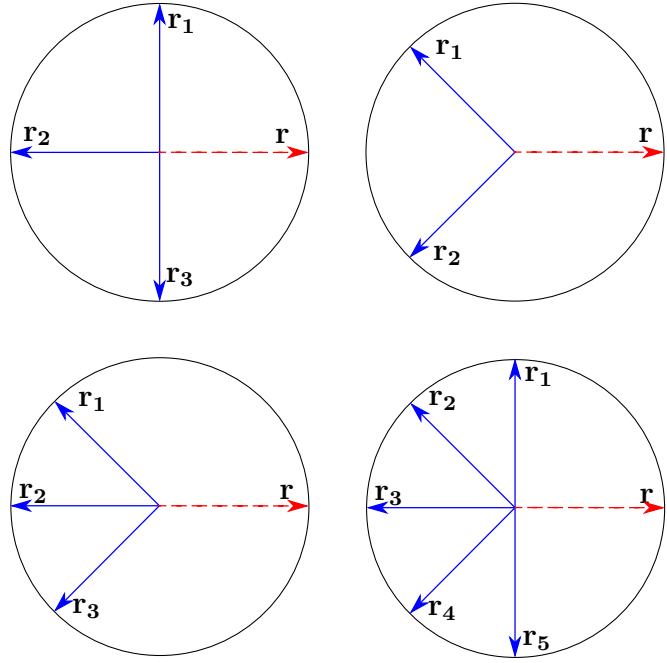
\captionsetup{width=\linewidth}
\caption{Examples of Bloch vectors on a plane. The sets of blue vectors are not antidistinguishable by themselves, but are made antidistinguishable by adding the red dashed vector into the set.}\label{fig:2}
\end{figure}

\begin{proposition}\label{prop:qubit}
Let $P_1,\ldots,P_n$ be pure qubit states. 
They are antidistinguishable if and only if there are positive real numbers $t_1,\ldots,t_n$ such that 
\begin{align}\label{eq:sum=id2}
\sum_{j=1}^n t_jP_j = \id \, . 
\end{align}
\end{proposition}

\begin{proof}
The 'if' part is Cor. \ref{cor:id}, thus it remains to prove the 'only if' part.
Let us assume that $P_1,\ldots,P_n$ are antidistinguishable, hence there exists a POVM $\Mo$ such that $P_j\Mo(j) = 0$.
This implies that $P_j^\perp \Mo(j) = \Mo(j)$, thereby $\Mo(j) \leq P_j^\perp$. Since $P_j^\perp$ is a one-dimensional projection, we conclude that $\Mo(j)= t_j P_j^\perp$ for some real number $t_j >0$.
Since 
\begin{equation}
\id=\sum_j \Mo(j)=\sum_j t_j P_j^\perp \, , 
\end{equation}
 we first see, by taking the trace on both sides, that $\sum_j t_j =2$.
 Then
\begin{equation}
\id = \sum_j t_j P_j^\perp = 2 \id - \sum_j t_j P_j \, ,
\end{equation}
which means that $\sum_j t_j P_j = \id$.
\end{proof}

A qubit state $\varrho$ can be written in the Bloch form $\varrho = \half (\id + \vr \cdot \vsigma)$, where $\vsigma=(\sigma_x,\sigma_y,\sigma_z)$, $\vr\in\real^3$, $\no{\vr}\leq 1$, and $\varrho$ is pure exactly when $\no{\vr}=1$.
Using Bloch vectors we can reformulate Prop. \ref{prop:qubit} as stated in \cite{CaFuSc02}.
\emph{A set $\{\half (\id + \vr_1 \cdot \vsigma),\ldots,\half (\id + \vr_n \cdot \vsigma)\}$  of pure qubit states is antidistinguishable if and only if there are positive real numbers $t_1,\ldots,t_n$ such that}
\begin{align}\label{eq:sum=0}
\sum_{j=1}^n t_j \vr_j = 0 \, . 
\end{align}
Namely, if these numbers $t_j$ exist, then we can scale them to sum up to $2$, and therefore \eqref{eq:sum=0} is equivalent to our earlier condition.

From Prop. \ref{prop:qubit} we get an improved version of Prop. \ref{prop:2n} for the qubit case. 

\begin{proposition}\label{prop:qubit-2}
Any finite set of pure qubit states is either antidistinguishable or can be made antidistinguishable by adding a single pure state. 
\end{proposition}

\begin{proof}
Let us assume that the set $\{\half (\id + \vr_1 \cdot \vsigma),\ldots,\half (\id + \vr_n \cdot \vsigma)\}$ of pure qubit states is not antidistinguishable, in which case $\vr\equiv\sum_j \vr_j \neq 0$.
We define $\vr_{n+1} = - \vr / \no{\vr}$. 
Then $\vr_{n+1}$ is a unit vector and 
\begin{equation}\label{eq:rn1}
\sum_{j=1}^{n} \frac{1}{\no{\vr}} \vr_j + \vr_{n+1} = 0 \, .
\end{equation}
The remaining step is to see that $\vr_{n+1} \neq \vr_j$.
Let us make a counter assumption that $\vr_{n+1} = \vr_j$ for some $j$. 
However, from \eqref{eq:rn1} it follows that
\begin{equation}
(1+\frac{1}{\no{\vr_j}}) \vr_j + \sum_{i \neq j} \frac{1}{\no{\vr_i}} \vr_i  = 0 \, , 
\end{equation}
which contradicts the fact that the original set is not antidistinguishable.
\end{proof}

Some examples of the completion process described in the proof of Prop. \ref{prop:qubit-2} are shown in Fig. \ref{fig:2}.

\section{Acknowledgement}

This work was performed as part of the Academy of Finland Centre of Excellence program (project 312058).
We thank Laura Piispanen for comments on an earlier version of the paper.

\end{document}

%% file: made_antidistinguishable.pdf_tex
\begingroup%
  \makeatletter%
  \providecommand\color[2][]{%
    \errmessage{(Inkscape) Color is used for the text in Inkscape, but the package 'color.sty' is not loaded}%
    \renewcommand\color[2][]{}%
  }%
  \providecommand\transparent[1]{%
    \errmessage{(Inkscape) Transparency is used (non-zero) for the text in Inkscape, but the package 'transparent.sty' is not loaded}%
    \renewcommand\transparent[1]{}%
  }%
  \providecommand\rotatebox[2]{#2}%
  \ifx\svgwidth\undefined%
    \setlength{\unitlength}{291.96850394bp}%
    \ifx\svgscale\undefined%
      \relax%
    \else%
      \setlength{\unitlength}{\unitlength * \real{\svgscale}}%
    \fi%
  \else%
    \setlength{\unitlength}{\svgwidth}%
  \fi%
  \global\let\svgwidth\undefined%
  \global\let\svgscale\undefined%
  \makeatother%
  \begin{picture}(1,0.93203883)%
    \put(0,0){\includegraphics[width=\unitlength,page=1]{made_antidistinguishable.pdf}}%
    \put(0.28102733,0.85883922){\color[rgb]{0,0,0}\makebox(0,0)[lb]{\smash{$\mathbf{r_1}$}}}%
    \put(0.0948118,0.71381171){\color[rgb]{0,0,0}\makebox(0,0)[lb]{\smash{$\mathbf{r_2}$}}}%
    \put(0.2789405,0.53185761){\color[rgb]{0,0,0}\makebox(0,0)[lb]{\smash{$\mathbf{r_3}$}}}%
    \put(0.42027257,0.71106968){\color[rgb]{0,0,0}\makebox(0,0)[lb]{\smash{$\mathbf{r}$}}}%
    \put(0,0){\includegraphics[width=\unitlength,page=2]{made_antidistinguishable.pdf}}%
    \put(0.74191403,0.4028349){\color[rgb]{0,0,0}\makebox(0,0)[lb]{\smash{$\mathbf{r_1}$}}}%
    \put(0.61823988,0.37227088){\color[rgb]{0,0,0}\makebox(0,0)[lb]{\smash{$\mathbf{r_2}$}}}%
    \put(0.54847714,0.24923579){\color[rgb]{0,0,0}\makebox(0,0)[lb]{\smash{$\mathbf{r_3}$}}}%
    \put(0.89041349,0.24605108){\color[rgb]{0,0,0}\makebox(0,0)[lb]{\smash{$\mathbf{r}$}}}%
    \put(0,0){\includegraphics[width=\unitlength,page=3]{made_antidistinguishable.pdf}}%
    \put(0.61720694,0.10083124){\color[rgb]{0,0,0}\makebox(0,0)[lb]{\smash{$\mathbf{r_4}$}}}%
    \put(0.73933155,0.06254836){\color[rgb]{0,0,0}\makebox(0,0)[lb]{\smash{$\mathbf{r_5}$}}}%
    \put(0,0){\includegraphics[width=\unitlength,page=4]{made_antidistinguishable.pdf}}%
    \put(0.6273211,0.82590996){\color[rgb]{0,0,0}\makebox(0,0)[lb]{\smash{$\mathbf{r_1}$}}}%
    \put(0.62428289,0.55994902){\color[rgb]{0,0,0}\makebox(0,0)[lb]{\smash{$\mathbf{r_2}$}}}%
    \put(0.88922928,0.71183964){\color[rgb]{0,0,0}\makebox(0,0)[lb]{\smash{$\mathbf{r}$}}}%
    \put(0,0){\includegraphics[width=\unitlength,page=5]{made_antidistinguishable.pdf}}%
    \put(0.16606347,0.37058429){\color[rgb]{0,0,0}\makebox(0,0)[lb]{\smash{$\mathbf{r_1}$}}}%
    \put(0.08938111,0.25170104){\color[rgb]{0,0,0}\makebox(0,0)[lb]{\smash{$\mathbf{r_2}$}}}%
    \put(0.42578179,0.25128428){\color[rgb]{0,0,0}\makebox(0,0)[lb]{\smash{$\mathbf{r}$}}}%
    \put(0,0){\includegraphics[width=\unitlength,page=6]{made_antidistinguishable.pdf}}%
    \put(0.15811091,0.10329649){\color[rgb]{0,0,0}\makebox(0,0)[lb]{\smash{$\mathbf{r_3}$}}}%
    \put(0,0){\includegraphics[width=\unitlength,page=7]{made_antidistinguishable.pdf}}%
  \end{picture}%
\endgroup%